\documentclass[a4paper,11pt]{article}
\usepackage{a4wide}
\usepackage{authblk}
\usepackage[colorlinks=true, linkcolor=blue, citecolor=blue, urlcolor=blue]{hyperref}
\usepackage{graphicx} 
\usepackage{my_package}
\usepackage[numbers,sort&compress]{natbib}
\usepackage[a4paper,margin=2.5cm]{geometry}

\usetikzlibrary{arrows.meta,positioning,calc}

\allowdisplaybreaks[3]

\newcommand{\beqeeq}[1]{
\begin{equation}
#1
\end{equation}
}

\title{The rate of purification of quantum trajectories}
\author[1,4]{Ma\"{e}l Bompais\thanks{Corresponding author: mael.bompais@nottingham.ac.uk}}
\author[2]{Nina H. Amini}
\author[3,4]{Juan P. Garrahan} 
\author[1,4]{M\u{a}d\u{a}lin Gu\c{t}\u{a}}

\affil[1]{School of Mathematical Sciences, University of Nottingham, University Park, Nottingham, NG7 2RD, United Kingdom}
\affil[2]{Université Paris-Saclay, CNRS, CentraleSupélec, Laboratoire des signaux et systèmes (L2S), France}
\affil[3]{School of Physics and Astronomy, University of Nottingham, University Park, Nottingham, NG7 2RD, United Kingdom}
\affil[4]{Centre for the Mathematics and Theoretical Physics of Quantum Non-equilibrium Systems, University of Nottingham, University Park, Nottingham, NG7 2RD, United Kingdom}

\date{}

\begin{document}

\maketitle


\begin{abstract}
{We investigate the behavior of quantum trajectories conditioned on measurement outcomes. Under a condition related to the absence of so-called dark subspaces,
K\"{u}mmerer and Maassen showed that such trajectories almost surely purify in the long
run. An alternative, simple proof of this result is first provided using Lyapunov methods. The result is then strengthened by showing that purification occurs at an exponential rate in expectation, again via a Lyapunov approach. The quantum filter stability 
problem is then considered: two trajectories are propagated under the same measurement
record -- one from the true initial state, the other from an arbitrary initial
guess -- and it is shown that the estimated trajectory converges exponentially fast to the
true one. This quantifies the rate at which information about the current state is
progressively revealed through the measurement process.}
\end{abstract}

\section{Introduction}
Quantum trajectories are random processes that describe the evolution of quantum systems subject to repeated indirect measurements \cite{wiseman2009quantum,gardiner2004quantum,carmichael1993quantum}. They are widely used in quantum optics and open systems theory \cite{carmichael1993quantum,gardiner2004quantum}, quantum control and filtering \cite{belavkin1992quantum,belavkin1983on-the-theory,belavkin1987non-demolition,bouten2009a-discrete,van-handel2005modelling}, and quantum information theory \cite{jacobs2014quantum,nielsen2000quantum,nielsen2001quantum}, where they model open quantum systems
under continuous monitoring. Across all these settings, the quantum trajectory formalism reveals how measurement backaction intertwines with the system’s own dynamics, and underlies many results in state estimation, feedback control, and decoherence. Understanding the long-term behavior and stability of such trajectories is therefore fundamental for both theoretical analysis and practical implementations of quantum technologies.

In an indirect measurement scheme, a system of interest interacts sequentially with a series of ancillary quantum systems, called probes. After each interaction, a direct measurement is performed on the probe. Due to the entanglement between system and probe, the measurement outcome provides partial information about the system, while also inducing a backaction on its state. Repeating this procedure with independent probes generates a stochastic sequence of states for the system—known as a quantum trajectory.

A natural question is how informative the measurement process is about the system's state. In \cite{maassen2006purification}, Maassen and Kümmerer showed that, for perfect measurements and under a structural condition—namely the absence of dark subspaces—the quantum trajectory purifies over time, in the sense that, starting from a potentially mixed state, it asymptotically reaches the set of pure states (while possibly continuing to evolve on this set).

Although purification is a property of the physical trajectory itself, it also has direct consequences for state estimation. From an estimation perspective, this purification behavior is closely related to the ability to reconstruct the current state of the trajectory from the measurement record when the initial state is unknown. Just as a direct measurement of a non-degenerate observable on an arbitrary quantum system determines the post-measurement state from the measurement outcome alone, a purifying quantum trajectory allows one to reconstruct the system state asymptotically from the accumulated measurement record \cite{amini2021on-asymptotic, benoist2019invariant}.
This notion is usually referred to as the stability of estimated quantum trajectories (or quantum filters). When the true system state and an estimated state are driven by the same measurement record, stability means that the dependence on the initial estimate is asymptotically lost. 

{In the case of perfect measurements,  on which we focus in this article, purification has been shown to be necessary and sufficient for the asymptotic stability of estimated quantum trajectories \cite{bompais2024asymptotic,amini2021on-asymptotic}. For imperfect measurements, trajectories typically do not purify — the set of pure states is generally not invariant under the dynamics — yet some filter stability results still exist: when the underlying quantum channel is irreducible, an appropriate contractivity condition yields necessary and sufficient condition for stability \cite{amini2025asymptotic}. These stability properties have also been used to establish the uniqueness of the invariant measure for both perfect and imperfect measurements \cite{benoist2019invariant,amini2025asymptotic}, and the set of invariant measures in the presence of dark subspaces was characterized in \cite{benoist2024dark}.}

While asymptotic purification and state-estimation results provide essential qualitative insight, they leave open quantitative questions regarding the speed of convergence. Understanding how fast purification and estimation occur is however crucial. At a conceptual level, convergence rates capture the efficiency of information transfer from the system to the detector, governing the reduction of uncertainty induced by measurement. From an applied perspective, these rates set fundamental limits on the performance of quantum filtering methods, the responsiveness of feedback control, and the feasibility of real-time state monitoring. Addressing these questions is therefore essential for both theoretical understanding and practical implementation.

In this paper, we address these quantitative questions using Lyapunov techniques. We first present a simple alternative proof of the Kümmerer–Maassen purification theorem based on the construction of a suitable Lyapunov function. Building on this approach, we then establish exponential purification in expectation, with an explicit convergence rate. Leveraging exponential purification, we further show that the fidelity between the true and estimated trajectories also converges exponentially fast in expectation, when both are driven by the same measurement record. This corresponds to exponential stability of the associated quantum filters. We note that a notion of exponential purification has also been investigated in \cite{benoist2024dark}. While related in spirit, our results rely on a different approach based on Lyapunov methods, leading to an explicit exponential rate.

The paper is organized as follows. Section \ref{sec:int} introduces the model of quantum trajectories together with the purification assumption, and gives an alternative proof of the K\"{u}mmerer–Maassen theorem. In Section \ref{sec:purity}, we prove exponential convergence in expectation for the purity. Section \ref{sec:stability} uses this result to derive exponential convergence of the fidelity between the true and estimated trajectories, strengthening previous results on filter stability. In Section \ref{sec:simulations}, we conduct a simulation study illustrating exponential purification in a spin chain model with measurements at one end. Section \ref{sec:conc} provides our conclusions.

\section{Quantum trajectories and  purification}\label{sec:int}
In this section we briefly present the setting of discrete-time quantum trajectories, introduce the purification assumption (i.e., the absence of dark subspaces), and  provide an alternative, simpler proof of the K\"{u}mmerer–Maassen purification result \cite{maassen2006purification}.

\subsection{Dynamics under repeated measurements}

Let us consider a quantum system described by a Hilbert space $\H$ of finite dimension $d.$ The state space of the system is the set of density operators
\beqeeq{
	\S = \set{ \rho \in \mathcal \B(\H) \mid \rho = \rho^\dag,~\rho \geq 0,~\tr{\rho}=1}.
}
Let us suppose that the system is subjected to indirect repeated measurements described in terms of Kraus operators $\{V_i\}_{i\in \I}$ on $\mathcal{H},$ labeled by a finite set of possible outcomes $\I$. The Kraus operators satisfy the  completeness relation \beqeeq{\sum_{i\in\I} V_i^\dag V_i = \mathbb I,} where $\mathbb I$ denotes the identity operator on $\H$. The evolution of the system's state is then described by a quantum trajectory, which is a Markov chain $(\rho_n)$ on the state space $\S$, with transitions between states given by 
\beqeeq{\label{eq:perfect_traj}
	\rho_{n+1} = \normalized{V_{i} \rho_n V_{i}^\dag} \qquad \text{with probability } \tr{V_{i} \rho_n V_{i}^\dag},\quad i \in \I.
}

By iterating the one-step update, if the observed measurement sequence up to time $n$ is $I=(i_1,\ldots,i_n)$, the conditional state of the system at time $n$, starting from $\rho_0$, is
\beqeeq{
    \rho_n = \normalized{V_{i_n} \cdots V_{i_1} \rho_0 V_{i_1}^\dag \cdots V_{i_n}^\dag} = \normalized{V_I \rho_0 V_I^\dag}{},
}
where we set the notation $V_I := V_{i_n} \cdots V_{i_1}$, corresponding to the product of the Kraus operators associated with the sequence $I$. The probability of observing this sequence is given by
\beqeeq{
    \mathbb P_{\rho_0}(I) = \tr{V_I \rho_0 V_I^\dag} = \tr{M_I \rho_0},
}
where $M_I := V_I^\dag V_I$. The operators $(M_I)_I$ define the POVM associated with sequences of measurement outcomes, satisfying $M_I \geq 0$ and $\sum_{I \in \I^n} M_I = \mathbb I$ for all $n\in \N.$ We denote by $\E_{\rho_0}$ the expectation with respect to $\P_{\rho_0}$. When there is no ambiguity, we drop the subscript and write $\P$ and $\E$ for $\P_{\rho_0}$ and $\E_{\rho_0}$ respectively. {The dynamics \eqref{eq:perfect_traj} corresponds to the so-called perfect measurement setting. More general state transformations arise in the imperfect measurement setting \cite{somaraju2012design}, for instance when measurement outcomes are coarse-grained. Trajectories typically do not purify in that case; hence, in this paper, we focus on the perfect measurement setting.}

\subsection{Purification assumption}\label{subsec:purification_assumption}

The \emph{purity} of a state $\rho$, which will be the quantity of interest in
this work, is defined by
\beqeeq{
    \mathrm{Purity} = \tr{\rho^{2}}.
}
It satisfies
${1}/{d} \;\le\; \tr{\rho^{2}} \;\le\; 1,$
with the lower bound attained by the maximally mixed state $\rho = \mathbb I/d$,
and the upper bound by pure states (rank-one projections). We say that a quantum trajectory \emph{almost surely purifies} when:
\beqeeq{  
    \lim_{n\to\infty} \tr{\rho_n^{2}} = 1 \quad \text{a.s.}
}
According to K\"{u}mmerer and Maassen \cite{maassen2006purification}, this occurs precisely when the collection of Kraus operators admits
no \emph{dark subspaces}. The precise statement of the assumption is as follows. We refer to it as the purification assumption.\\

\noindent\textbf{(Pur):~}
    For all projections $\pi$ with $\rank (\pi) \geq 2$, there exist $p\in \N$ and $(i_1,\ldots,i_p) \in \I^p$ such that
    $$\pi V_{i_1}^\dag \cdots V_{i_p}^\dag V_{i_p} \cdots V_{i_1}  \pi \not\propto \pi.$$

Here and throughout, $A \propto \pi$ means that $A = \lambda \pi$ for some
scalar $\lambda \in \C$.

When \pur does not hold, i.e., when there exists a projection $\pi$ with
$\rank(\pi) \ge 2$ such that
\beqeeq{\label{eq:prop_relation_dark}
    \pi V_{i_1}^\dag \cdots V_{i_p}^\dag V_{i_p} \cdots V_{i_1} \pi \propto \pi
}
for all $p \in \mathbb N$ and all $(i_1,\ldots,i_p)\in\I^p$,
the subspace $\range(\pi)$ is called a dark subspace.
A quantum trajectory initially supported on such a subspace has a constant purity, and, in particular, never purifies if the initial state is mixed. Ruling out the existence of such subspaces essentially amounts to requiring that the products of Kraus operators never act as an isometry on any non-trivial subspace of the Hilbert space.

The following theorem by K\"ummerer and Maassen~\cite{maassen2006purification}
states that hypothesis \pur is a sufficient condition for the almost sure purification of
quantum trajectories.

\begin{theorem}[K\"ummerer and Maassen, 2006]\label{thm:KM}
    Suppose that \pur holds. Then the quantum trajectory almost surely purifies.
\end{theorem}

The next proposition formalizes that, when a dark subspace exists, the purity indeed remains constant along trajectories whose initial state is supported on it. Thus, \pur is also a necessary condition for almost sure purification to hold for all initial states.

\begin{proposition}[Constancy of the purity on a dark subspace]\label{prop:purity_constant_in_dark_subspace}
Let $D \subset \H$ be a dark subspace with projection $\pi_D$, in the sense that
for every finite record $I=(i_1,\dots,i_p)$,
\beqeeq{
    \pi_D\,V_I^\dagger V_I\,\pi_D = \mu_I\,\pi_D,
}
for some scalar $\mu_I.$
Assume that the initial state $\rho_0$ is entirely supported on $D$.
Then the purity is conserved: for all $n\ge 0$,
\beqeeq{
    \tr{\rho_n^2} = \tr{\rho_0^2}.
}
\end{proposition}

\begin{proof} Since $\rho_0$ is supported on $D$, we have \beqeeq{ \tr{V_I\rho_0 V_I^\dagger} = \tr{V_I^\dagger V_I\rho_0} = \tr{\mu_I\pi_D\rho_0} = \mu_I. } On $D$, the relation $V_I^\dagger V_I = \mu_I \pi_D$ implies the existence of a unitary operator $U_I$ such that \beqeeq{ V_I\big|_D = \sqrt{\mu_I}\,U_I\big|_D. } Therefore, the normalized post-measurement state reads \beqeeq{ \rho_n = \frac{V_I\rho_0 V_I^\dagger}{\tr{V_I\rho_0 V_I^\dagger}} = U_I\rho_0 U_I^\dagger. } Since the purity is invariant under unitary conjugation, it follows that \beqeeq{ \tr{\rho_n^2} = \tr{\rho_0^2}. } \end{proof}

\medskip

Verifying hypothesis \pur is generally difficult in practice. 
The following proposition provides a practical criterion, showing that it suffices to check darkness up to a finite order. It will also be used in the proof of exponential purification in Section~\ref{sec:purity}.
Although this result is known within the
community \cite{privatecomm} (albeit not written explicitly), we include a proof for
completeness.

{
To this end, we introduce the \textit{observable space}
\beqeeq{
    \O
    =
    \spn
    \left\{
        V_{i_1}^\dag \cdots V_{i_p}^\dag V_{i_p} \cdots V_{i_1}
        \;\middle|\; (i_1,\dots,i_p)\in\I^p,\; p\in\N
    \right\}.
}
This space coincides with the linear span of the effects $(M_I)_I$ governing the output
probability law. In particular, the statistics of arbitrarily long measurement records
determine the quantities $\tr{X\rho_0}$ for all $X \in \O$. Hence, $\O$ captures
exactly the components of the initial state $\rho_0$ that can be identified from the statistics of the measurement process, while the orthogonal directions to $\O$ remain undetectable. 

We further consider the sequence of subspaces $(\O_p)_{p\in\N}$ of $\B(\H)$ defined by
\beqeeq{
    \O_p
    =
    \spn
    \left\{
        V_{i_1}^\dag \cdots V_{i_p}^\dag V_{i_p} \cdots V_{i_1}
        \;\middle|\; (i_1,\dots,i_p)\in\I^p
    \right\},
}
so that $\O = \bigcup_{p\in\N} \O_p$. Using the completeness relation
$\sum_{i\in\I} V_i^\dag V_i = \mathbb I$, we have $\O_p \subset \O_{p+1}$. Since each
$\O_p$ is a subspace of $\B(\H)$, whose dimension is $d^2$, the sequence
$(\dim \O_p)$ can increase at most $d^2$ times. Hence, there exists $p \le d^2$ such
that $\O_p = \O_{p+1};$ we denote by $\bar p$ the smallest such index. Using the iterative description
\beqeeq{
    \O_{p+1}
    =
    \spn
    \left\{
        V_i^\dag X V_i \;\middle|\; X\in \O_p,\ i\in\I
    \right\},
}
we obtain
\beqeeq{
    \O_{\bar p+2}
    =
    \spn
    \left\{
        V_i^\dag X V_i \;\middle|\; X\in \O_{\bar p+1},\ i\in\I
    \right\}
    =
    \spn
    \left\{
        V_i^\dag X V_i \;\middle|\; X\in \O_{\bar p},\ i\in\I
    \right\}
    =
    \O_{\bar p+1},
}
where we used $\O_{\bar p+1}=\O_{\bar p}.$ By induction,
$$\O_{\bar p} = \O_{\bar p+1} = \O_{\bar p+2} = \cdots,$$
and therefore 
$$\O_{\bar p} = \O.$$
This stabilization property is the key to reducing
the verification of \pur to finitely many steps, as made precise by the following proposition.

\begin{proposition}\label{prop:verif_ND_finite_length}
Hypothesis \pur holds if and only if for any projection $\pi$ with $\rank(\pi)\ge 2$
there exist a length $p \le \pbar$ and a word $(i_1,\ldots,i_p) \in\I^p$ such that
\beqeeq{
    \pi V_{i_1}^\dag \cdots V_{i_p}^\dag V_{i_p} \cdots V_{i_1} \pi \not\propto \pi.
}
\end{proposition}

\begin{proof}
For a projection $\pi \in \B(\H)$, define
$A_\pi = \{ X \in \B(\H) \mid \pi X \pi \propto \pi \}$.
By definition of \pur,
\beqeeq{
    \pur \;\Leftrightarrow\;
    \forall\ \pi,\ \rank(\pi)\ge 2,\ 
    \exists X \in \O \text{ such that } X \notin A_\pi.
}
Since $\O = \O_{\bar p}$, the claim follows.
\end{proof}
}

Hence, for a given subspace $\range(\pi)$ of dimension greater than one, if the proportionality relation \eqref{eq:prop_relation_dark} holds for all sequences of length up to $\bar p$, then it necessarily holds for sequences of length $p>\bar p$, and the subspace is indeed dark.
Conversely, if the subspace is not dark, there exists a sequence of length $p \le \pbar$ that violates the proportionality relation \eqref{eq:prop_relation_dark}. {We note, however, that $\pbar$ is a bound and that in some cases non-darkness may be detected at a length strictly smaller than $\pbar$ for all subspaces; see Section~\ref{sec:simulations} for an explicit example.}



\subsection{An alternative proof of Theorem \ref{thm:KM}}\label{subsec:alternative_proof_Kummerer_and_Maassen}

In this subsection, we provide a simple proof of the K\"{u}mmerer--Maassen purification result using Lyapunov techniques, while the original proof  \cite{maassen2006purification} is based on martingale methods and a fundamental inequality due to Nielsen \cite{nielsen2001characterizing}. Let us introduce the following candidate Lyapunov function:
\begin{equation}
    \label{eq:lyapunov.function}
    \mathcal V(\rho) = \sqrt{1 - \tr{\rho^2}}, \quad \rho \in \S.
\end{equation}
Note that $\mathcal V(\rho) \ge 0$ for all $\rho \in \S$, and $\mathcal V(\rho)=0$ if and only if $\rho$ is pure.
In addition, we have the following property.

\begin{proposition}\label{prop:V(rho)_supermartingale}
    The process $(\mathcal V(\rho_n))_{n\ge 0}$ is a supermartingale, that is,
for all $n\ge 0$ and all $p\in\mathbb N$,
\beqeeq{
    \EE{\mathcal V(\rho_{n+p})}{\rho_n} \le \mathcal V(\rho_n) \quad {\rm a.s.}
}
\end{proposition}

\begin{proof}
Let $\rho \in \S$ and let \(\norm{\cdot}\) denote the Frobenius norm, defined by \( \norm*{X} = \sqrt{\tr{X^\dag X}}\) for $X \in \B(\H).$ We have, for all \( n \) and all $p$,
\begin{align}
    \mathbb E [ \mathcal V(\rho_{n+p}) \vert \rho_n = \rho ] &= \sum_{I \in \I^p} \sqrt{1 - \frac{\tr{(V_I \rho V_I^\dag)^2}}{(\tr{V_I \rho V_I^\dag})^2}} \, \tr{V_I \rho V_I^\dag}
    \nonumber \\
    &= \sum_{I \in \I^p} \sqrt{(\tr{V_I \rho V_I^\dag})^2 - \tr{(V_I \rho V_I^\dag)^2}}
    \nonumber \\
    &= \sum_{I \in \I^p} \sqrt{ \oo \tr{V_I \rho V_I^\dag} - \sqrt{\tr{(V_I \rho V_I^\dag)^2}} \cc \oo \tr{V_I \rho V_I^\dag} + \sqrt{\tr{(V_I \rho V_I^\dag)^2}} \cc }
    \nonumber \\
    &\leq \sqrt{ 1 - \sum_{I \in \I^p} \sqrt{\tr{(V_I \rho V_I^\dag)^2}} } \, \sqrt{ 1 + \sum_{I \in \I^p} \sqrt{\tr{(V_I \rho V_I^\dag)^2}} }
    \nonumber \\
    &= \sqrt{ 1 - \oo \sum_{I \in \I^p} \sqrt{ \tr{(V_I \rho V_I^\dag)^2} } \cc^2 }
    \nonumber \\
    &= \sqrt{ 1 - \oo \sum_{I \in \I^p} \norm*{ \rho^{\frac{1}{2}} V_I^\dag V_I \rho^{\frac{1}{2}} } \cc^2 }
    \nonumber \\
    &\leq \sqrt{ 1 - \norm*{ \sum_{I \in \I^p} \rho^{\frac{1}{2}}   V_I^\dag V_I \rho^{\frac{1}{2}} }^2 }
    \nonumber \\
    &= \sqrt{ 1 - \norm*{\rho}^2 }
    \nonumber \\
    &= \mathcal V(\rho)
\end{align}
where the first inequality follows from the Cauchy--Schwarz inequality and the second from the triangle inequality. 
\end{proof}

The following proposition reproduces the K\"{u}mmerer–Maassen purification result.

\begin{proposition}[Almost sure purification]\label{prop:V(rho)_goes_to_zero_under_(Pur)}
    Assume that \pur holds. Then
    \beqeeq{
        \limn \mathcal V(\rho_n) = 0 \quad {\rm a.s.}
}
\end{proposition}

\begin{proof}
As we have seen in the proof of Proposition \ref{prop:V(rho)_supermartingale}, for all $n \ge 0$ and all $p \in \N$ we have almost surely
\begin{equation*}
    \EE{\mathcal V(\rho_{n + p})}{\rho_n} - \mathcal V(\rho_n) \leq \mathcal Q_p(\rho_n),
\end{equation*}
where
\beqeeq{
  \mathcal Q_{ p}(\rho) := \sqrt{ 1 - \oo \sum_{I \in \I^p} \norm{ \rho^{1/2} V_I^\dag V_I \rho^{1/2} } \cc^2 } 
  - \sqrt{1 - \norm{\rho}^2} \leq 0, \qquad \rho \in \S.
}
Let us denote by $\S^{\textrm{acc}}$ the (almost sure) set of accumulation points of $(\rho_n)$. 
By Theorem~\ref{thm:accumulation_points_are_zeros_of_the_increment} in the Appendix, we have
\beqeeq{
  \S^{\textrm{acc}} \subset \bigcap_{p \geq 1} \{ \rho \in \S ~|~ \mathcal Q_p(\rho)=0\}.
}
Recall that we obtained $\mathcal Q_p(\rho) \le 0$ for all $p \in \mathbb N$ and all $\rho \in \S$ by applying the triangle inequality,
\beqeeq{
  \left\| \sum_{I \in \I^p} \rho^{1/2} V_I^\dag V_I \rho^{1/2} \right\|
  \leq \sum_{I \in \I^p} \norm{\rho^{1/2} V_I^\dag V_I \rho^{1/2}}.
}
Equality $\mathcal Q_p(\rho)=0$ holds only if we have equality in the triangle inequality,
\beqeeq{
\Big\|\sum_{I\in\I^p} \rho^{1/2}V_I^\dag V_I\rho^{1/2}\Big\|
    = \sum_{I\in\I^p}\|\rho^{1/2}V_I^\dag V_I\rho^{1/2}\|.
}
Since the Frobenius norm is strictly convex, this happens if and only if all the operators
$\rho^{1/2}V_I^\dag V_I\rho^{1/2}$ are nonnegative multiples of each other. As their sum is equal
to $\rho$, they must all be proportional to $\rho$, that is,
\beqeeq{
    \forall p \in \N,~~\forall I\in\I^p,\qquad
    \rho^{1/2}V_I^\dag V_I\rho^{1/2} \propto \rho.
}
Multiplying on the left and on the right by the square-root pseudo-inverse $\rho^{-1/2}$
(i.e.\ the inverse of $\rho^{1/2}$ on its support),
we finally obtain
\beqeeq{
  \forall p \in \N,~~\forall I \in \I^p, \qquad \pi V_I^\dag V_I \pi \propto \pi,
}
where \( \pi \) denotes the orthogonal projection onto \( \operatorname{range}(\rho^{1/2}) \). 
If $\rho$ has rank at least two, so does $\rho^{1/2}$, and hence $\pi$ as well, contradicting \pur. Therefore, the common zeros of all the functions $\mathcal Q_p,$ $p\in  \N,$ are precisely the rank-one operators, i.e.\ the pure states, and the trajectory $(\rho_n)$ accumulates only at pure states.

Since $(\mathcal V(\rho_n))_{n\ge 0}$ is a nonnegative supermartingale, it converges almost surely.
Let $\rho$ be any almost sure accumulation point of $(\rho_n)$. As shown above, every such $\rho$ 
must be pure, hence $\mathcal V(\rho)=0$. The almost sure limit of the real sequence $(\mathcal V(\rho_n))$ must therefore coincide with the 
value taken on these accumulation points, which forces 
\beqeeq{
    \lim_{n\to\infty} \mathcal V(\rho_n)=0 \quad \text{a.s.}
}

\end{proof}

\begin{remark}
When Assumption~$\pur$ fails, the common zeros of the functions $\mathcal Q_p$ are 
exactly the states $\rho$ supported on a dark subspace. In this situation, we recover the second alternative in the K\"{u}mmerer--Maassen theorem (see \cite[Section 5]{maassen2006purification}): when $\pur$ does not hold, the trajectory 
is asymptotically supported on a dark subspace at each time, and, by definition of 
dark subspaces, each Kraus operator sends a dark subspace to another through an 
isometry. As a consequence, the asymptotic dynamics is that 
of a random walk on the collection of dark subspaces, with unitary evolution along 
each transition. 
We refer to Ref.~\cite{benoist2024dark} for further discussion on the ergodic properties of this random walk between dark subspaces.
\end{remark}

\section{Exponential purification}\label{sec:purity}
In this section, we show that purification occurs exponentially fast in expectation. 
To this end, we consider the Lyapunov function $\mathcal V(\rho) $ introduced in Equation \eqref{eq:lyapunov.function}. 
\subsection{Main theorem}
Before stating the main result, we introduce a few notations. We denote by $\mathfrak{B}$ the set of all orthonormal bases
$\boldsymbol{\psi} = (\psi_1,\dots,\psi_d)$ of the Hilbert space $\mathcal H$.
We define the set of weights
\beqeeq{
\mathcal W
 := \left\{\, w=(w_{kl})_{1 \leq k<l \leq d} ~\big|~ 
\exists\,p\in\Delta_{d-1},\ \sum_{j=1}^d p_j^2<1,\
w_{kl} = \frac{2 p_k p_l}{1-\sum_j p_j^2}
\,\right\}.
}
Here $\Delta_{d-1}$ denotes the set of probability distributions $p =(p_1,\dots , p_d)$. For any multi-index $I=(i_1,\dots , i_p)$ we recall that  
$M_I = V_{i_1}^\dagger \cdots V_{i_p}^\dagger V_{i_p}\cdots V_{i_1}$ and denote by $M_I|_{\psi_k,\psi_l}$ the restriction of $M_I$ to the two dimensional space spanned by  $(\psi_k,\psi_l)$. Finally, $\bar p$ denotes the finite length necessary to form the observable space $(\O_{\bar p} = \O)$ derived in Section~\ref{subsec:purification_assumption}.

\begin{theorem}[{Exponential purification}]\label{thm:exponential_purification}
Assume that \pur holds. Then the quantum trajectory $(\rho_n)$
purifies exponentially fast in expectation: for all $n\ge0$,
\beqeeq{
    \E{\mathcal V(\rho_n)}
    \le 
    \mathcal V(\rho_0)\, e^{-\gamma\lfloor n/{\pbar}\rfloor}
}
where
\beqeeq{
    \gamma =
        -\ln \oo \sup_{\substack{\boldsymbol w \in \mathcal{W},~
                \boldsymbol{\psi}\in\mathfrak{B}}}
 \oo
            \sum_{I\in\I^{\pbar}}
            \sqrt{ \sum_{k<l} w_{kl} \det\bigl( M_I|_{\psi_k,\psi_l} \bigr) } \cc \cc, 
}
with $\gamma >0.$
\end{theorem}
{
The theorem immediately yields the following tail bound.
\begin{corollary}[Exponential tail bound]\label{cor:exponential_tail_bound_purity}
Under the assumptions of Theorem~\ref{thm:exponential_purification}, for all $\varepsilon>0$ and all $n\ge 0$,
\beqeeq{
\mathbb P\big(\mathcal V(\rho_n)\ge \varepsilon\big)
\le
\frac{\mathcal V(\rho_0)}{\varepsilon}\, e^{-\gamma\lfloor n/{\pbar}\rfloor}.
}
\end{corollary}

\begin{proof}
By Markov's inequality,
\beqeeq{
\mathbb P\big(\mathcal V(\rho_n)\ge \varepsilon\big)
\le
\frac{\E{\mathcal V(\rho_n)}}{\varepsilon},
}
and the result follows from Theorem~\ref{thm:exponential_purification}.
\end{proof}
}

Before proceeding with the proof of the theorem, let us state a few remarks concerning the rate.
\begin{itemize}
\item For rank-one Kraus operators,
    the rate is equal to $\gamma = +\infty$, using the convention $\ln(0) = -\infty$.
    In this case, the quantum state becomes pure after one single step of time.
\item When Assumption~$\pur$ does not hold, so that there exists a dark subspace,
the rate is $\gamma = 0$, consistent with the fact that the purity remains constant for trajectories with initial states supported on a dark subspace.
    \item For a qubit, with $\dim \H = 2$, the rate simplifies to
    \beqeeq{
        \gamma = - \ln \oo \sum_{I \in \I^{\bar p}} \sqrt{\det (M_I)} \cc.
}
\end{itemize}

The rate $\gamma$ admits a natural geometric interpretation in terms of the contraction of
two--dimensional volumes induced by the Kraus operators, see Fig.~\ref{fig:contraction}.
For a given eigenbasis $(\psi_k)$ of $\rho$, each pair $(\psi_k,\psi_l)$ spans a
two--dimensional subspace of the Hilbert space, and the quantity
$\det\bigl(M_I|_{\psi_k,\psi_l}\bigr)$ measures the contraction of the oriented area
generated by this pair under the action of the operator $M_I$.
The weights $w_{kl}$ encode the contribution of each such two--dimensional component, so that
the sum $\sum_{k<l} w_{kl}\det\bigl(M_I|_{\psi_k,\psi_l}\bigr)$ represents a weighted average
of area contractions over all pairs of eigenvectors.

The summation over all words $I\in\I^{\bar p}$ captures the cumulative effect of the
measurement channel over blocks of length $\bar p$, while the supremum over admissible
weights and orthonormal bases extracts the worst--case contraction scenario. The strict positivity of $\gamma$ is precisely what yields the exponential decay of
$\E{\mathcal V(\rho_n)}$ established in
Theorem~\ref{thm:exponential_purification}.\\

\begin{figure}[htbp]
    \centering
    \begin{tikzpicture}[scale=1.2]
        \begin{scope}[shift={(0,0)}]
            \node at (-2,2) {\large (a)};
            \draw[->, thick] (0,0) -- (2,0) node[right] {$\psi_k$};
            \draw[->, thick] (0,0) -- (0,2) node[above] {$\psi_l$};
            \fill[blue!20, opacity=0.4]
                (0,0) -- (2,0) -- (2,2) -- (0,2) -- cycle;
            \draw[->, thick] (3,1) -- (4.3,1) node[midway, above] {$V_I$};
            \begin{scope}[shift={(5.2,0)}]
                \draw[->, thick] (0,0) -- (1.2,0) node[midway, below] {$V_I\psi_k$};
                \draw[->, thick] (0,0) -- (0,1.2) node[midway, left] {$V_I\psi_l$};
                \fill[red!20, opacity=0.4]
                    (0,0) -- (1.2,0) -- (1.2,1.2) -- (0,1.2) -- cycle;
                \node at (2.3,1.7)
                    {\small $\det\bigl(M_I|_{\psi_k,\psi_l}\bigr)$};
                \draw[thick] (1.05,1.05) -- (1.9,1.55);
            \end{scope}
        \end{scope}
        \begin{scope}[shift={(0,-3)}]
            \node at (-2,2) {\large (b)};
            \draw[->, thick] (0,0) -- (2,0) node[right] {$\psi_k$};
            \draw[->, thick] (0,0) -- (0,2) node[above] {$\psi_l$};
            \fill[blue!20, opacity=0.4]
                (0,0) -- (2,0) -- (2,2) -- (0,2) -- cycle;
            \draw[->, thick] (3,1) -- (4.3,1) node[midway, above] {$V_I$};
            \begin{scope}[shift={(5.2,0)}]
                \draw[->, thick] (0,0) -- (1.4,0.15) node[right] {$V_I\psi_k$};
                \draw[->, thick] (0,0) -- (0.45,1.0) node[above] {$V_I\psi_l$};
                \fill[red!20, opacity=0.4]
                    (0,0)
                    -- (1.4,0.15)
                    -- (1.85,1.15)
                    -- (0.45,1.0)
                    -- cycle;
                \node at (2.6,1.55)
                    {\small $\det\bigl(M_I|_{\psi_k,\psi_l}\bigr)$};
                \draw[thick] (1.0,0.75) -- (1.9,1.25);
            \end{scope}
        \end{scope}

    \end{tikzpicture}
    \caption{
        (a) Geometric action of a Kraus operator $V_I$ on the two--dimensional subspace $\spn(\psi_k,\psi_l)$ when this subspace is dark. For any word $I$, the images $V_I\psi_k$ and $V_I\psi_l$ remain orthogonal and have the same norm. 
        (b) Geometric action of a Kraus word $V_I$ on the two--dimensional subspace $\spn(\psi_k,\psi_l)$ in the non--dark case. There exists at least one word $I$ such that the images $V_I\psi_k$ and $V_I\psi_l$ are not orthogonal and/or do not have the same norm, leading to a strict distortion of the associated two--dimensional area.
        }
    \label{fig:contraction}
\end{figure}
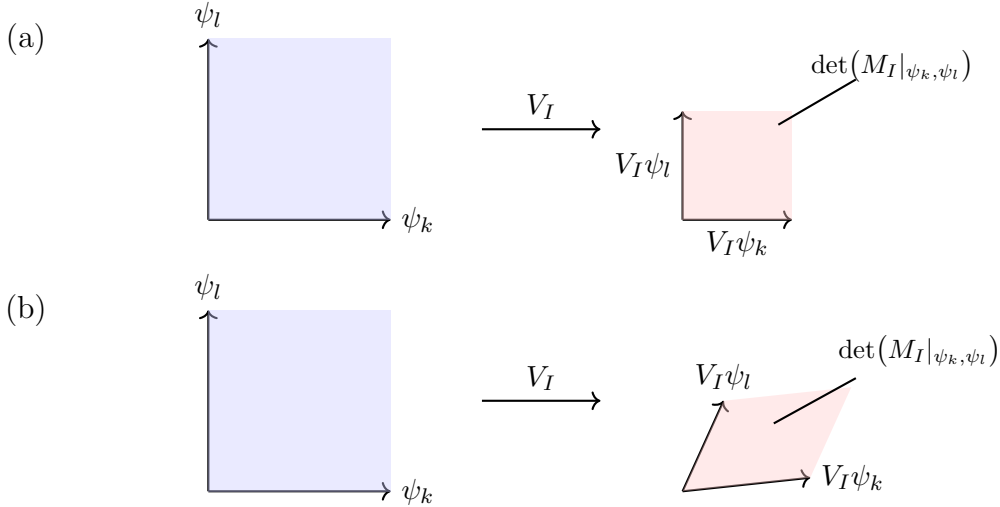

We now proceed with the proof of the theorem.

\begin{proof}[Proof of Theorem \ref{thm:exponential_purification}]

The proof consists in showing that there exists $\lambda < 1$ such that, for all $n,$
\begin{equation}\label{eq:strict_contraction}
    \EE{\mathcal V(\rho_{n+\bar p})}{\rho_n} \leq \lambda \mathcal V(\rho_n) \quad \as
\end{equation}
Then, taking expectation value on both sides will yield the result. Define 
\beqeeq{ \S_\mix = \set{ \rho \in \S \mid \tr{\rho^2} < 1 } \qquad ; \qquad\S_\pure = \set{ \rho \in \S \mid \tr{\rho^2} =1}.} For $\rho \in \S_\mix$ (so that $\V(\rho)>0$), we may write
\beqeeq{
	\EE{\V(\rho_{n+\bar p})}{\rho_n=\rho} = \frac{\EE{\V(\rho_{n+\bar p})}{\rho_n=\rho}}{\V(\rho)} \V(\rho).
}
We then define, for all $\rho \in \S_\mix,$
\beqeeq{
	\lambda(\rho) : = \frac{\EE{\V(\rho_{n+\bar p})}{\rho_n = \rho}}{\V(\rho)}.
}
From the proof of Proposition \ref{prop:V(rho)_goes_to_zero_under_(Pur)}, the Lyapunov
function is strictly decreasing in expectation on $\S_\mix$ under $\pur$, and therefore
\begin{equation}\label{eq:lambda_rho_<1}
    \forall \rho \in \S_\mix,~~ \lambda(\rho)<1.
\end{equation}
Yet our aim is to show that \( \sup_{\rho \in \S_\mix} \lambda(\rho) < 1, \) which would then yield Eq.~\eqref{eq:strict_contraction} with $\lambda <1.$ Since $\S_\mix$ is not compact, we cannot deduce a uniform gap from~$1$ directly. 
We proceed by contradiction and suppose that there exists a sequence of states
\beqeeq{
  \bigl(\rho^{(n)}\bigr)_{n \ge 0} \subset \mathcal{S}_{\mathrm{mix}}
  \quad \text{such that} \quad
  \lim_{n \to \infty} \lambda\!\bigl(\rho^{(n)}\bigr) = 1.
}
By compactness of the state space \( \mathcal{S} \), we may assume, up to extraction, that this sequence is convergent:
\beqeeq{
  \rho^{(n)} \xrightarrow[n \to \infty]{} \rho^* \in \overline{\mathcal{S}_{\mathrm{mix}}} = \mathcal{S}.
}
Necessarily, however, \( \rho^* \in \mathcal{S}_{\mathrm{pure}} \). Indeed, if \( \rho^* \in \S \setminus \S_{\pure} = \mathcal{S}_{\mathrm{mix}} \), then by continuity of $\lambda(\rho)$ on $\S_\mix$ we would have \( \lambda(\rho^*) = 1 \), which would contradict Equation \eqref{eq:lambda_rho_<1}.

To analyze the behavior of the coefficient $\lambda(\rho)$ near pure states, we decompose
any state $\rho \in \S_\mix$ as $\rho=\sum_k p_k \ketbra{\psi_k},$ for some
probability vector $\mathbf{p}=(p_k)$, with $\sum_k p_k^2 < 1$, and some
orthonormal basis $\boldsymbol{\psi}=\{\ket{\psi_k}\}$ of $\H$. We have

\begin{align}
    \mathbb E [ \mathcal V(\rho_{n+\bar p}) \vert \rho_n = \rho ] &= \sum_{I \in \I^{\bar p}} \sqrt{1-\frac{\tr{\paren{V_I \rho {V_I}^\dag}^2}}{(\tr{V_I \rho {V_I}^\dag})^2}}\tr{V_I \rho {V_I}^\dag} \nonumber\\
    &= \sum_{I \in \I^{\bar p}} \sqrt{(\tr{V_I \rho {V_I}^\dag})^2 - \tr{(V_I \rho {V_I}^\dag)^2}} \nonumber\\
    &= \sum_{I \in \I^{\bar p}} \sqrt{ \sum_{k, l} p_k p_l \bra{\psi_k} M_I \ket{\psi_k} \bra{\psi_l} M_I \ket{\psi_l}  - \sum_{k,l} p_k p_l \abs{\bra{\psi_k} M_I \ket{\psi_l}}^2 } \nonumber\\
    &= \sum_{I \in \I^{\bar p}} \sqrt{ \sum_{k \neq l} p_k p_l \oo\bra{\psi_k} M_I \ket{\psi_k} \bra{\psi_l} M_I \ket{\psi_l}  - \abs{\bra{\psi_k} M_I \ket{\psi_l}}^2 \cc} \nonumber\\
    &= \frac{\sum_{I \in \I^{\bar p}} \sqrt{ \sum_{k \neq l} p_k p_l \oo\bra{\psi_k} M_I \ket{\psi_k} \bra{\psi_l} M_I \ket{\psi_l}  - \abs{\bra{\psi_k} M_I \ket{\psi_l}}^2 \cc}}{\sqrt{1 - \sum_j p_j^2} }  \mathcal V(\rho) \nonumber\\
    &= \sum_{I \in \I^{\bar p}}\sqrt{\sum_{k < l} \frac{2\,p_k p_l}{1-\sum_j p_j^2} \oo\bra{\psi_k} M_I \ket{\psi_k} \bra{\psi_l} M_I \ket{\psi_l}   - \abs{\bra{\psi_k} M_I \ket{\psi_l}}^2  \cc}  \mathcal V(\rho) \nonumber\\
    &= \sum_{I \in \I^{\bar p}} \sqrt{ \sum_{k < l } w_{kl} \det (M_I \vert_{\psi_k,\psi_l})} \V(\rho)
\end{align}
with $w_{kl} = \frac{2 p_k p_l}{1 - \sum_j p_j^2}.$
Note that we then have
\beqeeq{
	\lambda(\rho) = \sum_{I \in \I^{\bar p}} \sqrt{ \sum_{k < l } w_{kl} \det (M_I \vert_{\psi_k,\psi_l})}
}
and 
\beqeeq{
	\sup_{\rho \in \S_\mix} \lambda(\rho) = \sup_{ w \in \W, \boldsymbol \psi \in \mathfrak B} \sum_{I \in \I^{\bar p}} \sqrt{ \sum_{k < l } w_{kl} \det (M_I \vert_{\psi_k,\psi_l})}.
}
Let us then write
\beqeeq{
  \rho^{(n)} = \sum_k p_k^{(n)} \ketbra{\psi_k^{(n)}}
}
as a convex decomposition into pure states of the sequence $(\rho^{(n)})$ achieving the supremum.  
Since \( \rho^{(n)} \to \rho^* =: \ketbra{\psi_1} \in \mathcal{S}_{\mathrm{pure}} \), we may assume, up to relabeling and extraction, that 
\beqeeq{
  \ketbra{\psi_1^{(n)}} \to \ketbra{\psi_1}
  \quad \text{and} \quad
  p_1^{(n)} \to 1 \quad \text{as } n \to \infty.
}
Writing $w_{kl}^{(n)} = \frac{2 p_k^{(n)} p_l^{(n)}}{1 - \sum_j (p_j^{(n)})^2}$ we then have, up to extraction again,
\beqeeq{
  w_{kl}^{(n)} \xrightarrow[n \to \infty]{} 
  \begin{cases}
    p_l' & \text{if } k = 1,\\[2mm]
    0 & \text{if } k \neq 1
  \end{cases}
}
where {\( p_l' = \limn \dfrac{2 p_l^{(n)}}{1 - \sum_j (p_j^{(n)})^2} \)}.  
Moreover, still up to extraction, we may assume that for all \( l = 2, \dots, d \), 
\beqeeq{
  \ket{\psi_l^{(n)}} \xrightarrow[n \to \infty]{} \ket{\psi_l}
}
for some orthogonal states $\{\psi_l\}.$ In the following, we will denote for simplicity $M_I^{ll} := \bra{\psi_l}  M_I  \ket{\psi_l}.$
The limit $\lambda^*$ of the coefficient $\lambda(\rho^{(n)})$ then satisfies
\begin{align}
\lambda^* &= \sum_{I \in \I^{\bar p}} \sqrt{ \sum_{l\neq 1} p_l^\prime \det ( M_I|_{\psi_l,\psi_1})} \nonumber\\
&= \sum_{I \in \I^{\bar p}} \sqrt{ \sum_{l\neq 1} p_l^\prime (M_I^{ll}M_I^{11} - \abs{M_I^{l1}}^2)} \nonumber\\
&\leq \sum_{I \in \I^{\bar p}} \sqrt{ \sum_{l\neq 1} p_l^\prime M_I^{ll}M_I^{11}} \nonumber\\
&= \sum_{I \in \I^{\bar p}} \sqrt{ \overline{M}_I M_I^{11}} \nonumber\\
& \leq \sqrt{\oo\sum_{I \in \I^{\bar p}} \overline{M}_I \cc} \sqrt{\oo\sum_{I \in \I^{\bar p}} M_I^{11} \cc} \nonumber\\
&=1
\end{align}
where we defined $\overline{M}_I := \sum_{l \neq 1} p_l^\prime M_I^{ll},$ satisfying $\sum_{I \in \I^\pbar} \overline{M}_I =1.$
Note that if $\abs{M_I^{l1}}\neq 0$ for some $l$ {with $p_l^\prime \neq 0,$} then the inequality is strict in the first inequality and $\lambda^*<1.$
We then suppose from now on that $M_I^{l1} = 0$ for all $l$ {such that $p_l^\prime \neq 0.$}

To have $\lambda^*=1,$ we must have equality in the Cauchy--Schwarz inequality, in which case there exists a scalar $\alpha\ge 0$
such that
\beqeeq{
    \forall\, I\in\I^\pbar, \quad M_I^{11} = \alpha \sum_{l\neq 1} p_l^\prime\, M_I^{ll}
    .
}
Summing over $I$ gives $\alpha =1$ and therefore
\begin{equation}\label{eq:link_between_M_I^00_and_M_I^kk}
    \forall\, I\in\I^\pbar, \quad M_I^{11}
    = \sum_{l\neq 1} p_l^\prime\, M_I^{ll}
    .
\end{equation}
{This relation can be rewritten in a more compact form by introducing the states
$$
\rho_0^{\prime} := \ketbra{\psi_1},
\qquad
\rho_0^{\prime\prime} := \sum_{l \neq 1} p_l^\prime\, \ketbra{\psi_l}.
$$
Then, the previous equality reads
\begin{equation}
    \forall I \in \I^{\pbar},~~\tr{M_I \rho_0^{\prime}} = \tr{M_I \rho_0^{\prime\prime}}.
\end{equation}
In particular, this implies that the associated probability measures coincide:
$$
\P_{\rho_0^{\prime}} = \P_{\rho_0^{\prime\prime}}.
$$
Let us consider now a trajectory starting from the state
$$
\rho_0 = \frac12 \rho_0^{\prime} + \frac12 \rho_0^{\prime\prime}.
$$
By equality of the probability measures, for any measurement record $I$ we have
$$
\tr{V_I \rho_0^{\prime} V_I^\dag}
= \tr{V_I \rho_0^{\prime\prime} V_I^\dag}
= \tr{V_I \rho_0 V_I^\dag}.
$$
Hence,
$$
\rho_n
= \frac{V_I \rho_0 V_I^\dag}{\tr{V_I \rho_0 V_I^\dag}}
= \frac12 \rho_n^{\prime} + \frac12 \rho_n^{\prime\prime},
$$
where
$$
\rho_n^{\prime} := \frac{V_I \rho_0^{\prime} V_I^\dag}{\tr{V_I \rho_0^{\prime} V_I^\dag}},
\qquad
\rho_n^{\prime\prime} := \frac{V_I \rho_0^{\prime\prime} V_I^\dag}{\tr{V_I \rho_0^{\prime\prime} V_I^\dag}}.
$$
Since $M_I^{l1} = 0$ for all $l$ {such that $p_l^\prime \neq 0$}, the operators $M_I$ {admit a block diagonal structure on the subspace}
$$\supp(\rho_0') \oplus \supp(\rho_0'')$$
for all words $I$ of length $\bar p$, and since any $M_I$ can be written as a linear combination of such operators (by the completeness relation for $|I|\le \bar p$ and by the fact that $\O=\O_{\bar p}$ for $|I| > \bar p$), it follows that $M_I$ is block diagonal {on $\supp(\rho_0') \oplus \supp(\rho_0'')$} for all $I$ of any length. Therefore,
$$
\tr{\rho_n' \rho_n''}
\propto
\tr{V_I \rho_0' V_I^\dag V_I \rho_0'' V_I^\dag}
=
\tr{\rho_0' M_I \rho_0'' M_I}
=
0,
$$
which shows that the supports of $\rho_n'$ and $\rho_n''$ remain orthogonal. Consequently,
$$
\tr*{\rho_n^2}
= \frac14 \tr*{(\rho_n^{\prime})^2}
+ \frac14 \tr*{(\rho_n^{\prime\prime})^2}.
$$
Since $\rho_n^{\prime}$ is pure, $\tr*{(\rho_n^{\prime})^2}=1$, and since $\rho_n^{\prime\prime}$ is a state, $\tr*{(\rho_n^{\prime\prime})^2}\leq 1$. Hence, for all $n,$
$$
\tr*{\rho_n^2}
\leq \frac12.
$$
This contradicts the almost sure purification of all trajectories, regardless of the initial state, when \pur holds.}

Therefore, there cannot exist a sequence of states $(\rho^{(n)}) \subset \S_\mix$ such that $\lambda(\rho^{(n)}) \to 1$ and we conclude
\beqeeq{
	 \sup_{\rho \in \S_\mix} \lambda(\rho) < 1.
}
Denoting $\lambda = \sup_{\rho \in \S_\mix} \lambda(\rho),$ we then have, almost surely,
$$
\EE{\V(\rho_{n+\bar p})}{\rho_n} \leq \lambda \V(\rho_n).
$$
Taking expectation on both sides and using the tower property yields
$$
\E{\V(\rho_{n+\bar p})}
= \E{\EE{\V(\rho_{n+\bar p})}{\rho_n}}
\le \lambda\, \E{\V(\rho_n)}.
$$
Iterating this inequality gives, for all $k\ge 0$,
$$
\E{\V(\rho_{k\bar p})} \le \lambda^k \V(\rho_0).
$$
Now write $n = k\bar p + r$ with $k=\lfloor n/\bar p\rfloor$ and $0\le r<\bar p$.
Since $(\mathcal V(\rho_n))$ is a supermartingale, we have $\E{\V(\rho_n)} \le \E{\V(\rho_{k\bar p})}$,
and applying the previous bound yields
$$
\E{\V(\rho_n)} \le \lambda^{\lfloor n/\bar p\rfloor}\,\V(\rho_0).
$$
Setting $\gamma = -\ln(\lambda) >0$, we finally obtain
$$
\E{\V(\rho_n)}
\le 
\V(\rho_0)\, e^{-\gamma\lfloor n/\bar p\rfloor}.
$$
\end{proof}

\subsection{Computing the rate bound for different lengths}

For each block length $p\ge1$, we introduce the associated $p$--step rate
\begin{equation}\label{eq:def_gamma_p}
    \gamma_{p}
    :=
    -\ln \Bigg(
        \sup_{\boldsymbol w , \boldsymbol{\psi}}
            \sum_{I\in\I^{ p}}
            \sqrt{\sum_{k < l} w_{kl} \det \big( M_I|_{\psi_k,\psi_l} \big)}
    \Bigg).
\end{equation}
For every $p\ge1$, the same calculations as in the previous section yield the bound
\begin{equation}\label{eq:exp_bound_with_arbitrary_p}
    \E{\mathcal V(\rho_n)}
    \le \mathcal V(\rho_0)\, e^{-\gamma_p\lfloor n/p\rfloor}.
\end{equation}
We can guarantee that $\gamma_p>0$ for $p\ge\bar p$,
where we recall $\bar p$ is the length given in Proposition~\ref{prop:verif_ND_finite_length}. Therefore, the bound is \textit{a priori} non-trivial from $\bar p$ onward. The rate $\gamma$ appearing in Theorem~\ref{thm:exponential_purification} is precisely
$\gamma = \gamma_{\bar p}$. Nevertheless, increasing the length $p$ on which $\gamma_p$ is computed can potentially improve the bound — at least at times $n$ that are multiples of $p$, the bound being constant in between consecutive multiples — as we shall see below.

First note that the rate $\gamma_p$ ($p \geq 1$) is expressed in terms of the contraction factor over $p$ steps,
\beqeeq{
    \gamma_p = - \ln (\lambda_p),
}
where
\beqeeq{
    \lambda_p(\rho)
    :=
    \frac{\EE{\mathcal V(\rho_p)}{\rho_0=\rho}}{\mathcal V(\rho)},
    \qquad 
    \lambda_p := \sup_{\rho\in\S_{\mix}} \lambda_p(\rho).
}

Below, we show that the sequence $(\lambda_p)$ is sub-multiplicative, which implies that
$(\gamma_p)$ is super-additive.

\begin{proposition}
\label{prop:lambdan}
The sequence $(\lambda_p)$ is sub-multiplicative. For all $p,q \in \N,$
\beqeeq{
\lambda_{p+q} \leq \lambda_p \lambda_q.
}
Consequently the following limit exists and satisfies
\beqeeq{
\lim_{p \to \infty} \lambda_p^{1/p} = \inf_{p \geq 1} \lambda_p^{1/p}.
}
\end{proposition}

\begin{proof}
Let $p,q \in \N$. Using the tower property of conditional expectation, we compute
\begin{align*}
\EE{\mathcal V(\rho_{p+q})}{\rho_0} 
&= \EE*{ \EE*{ \mathcal V(\rho_{p+q}) }{ \rho_p } }{ \rho_0 } \\
&= \EE*{ \frac{ \EE*{ \mathcal V(\rho_{p+q}) }{ \rho_p } }{ \mathcal V(\rho_p) } \cdot \mathcal V(\rho_p) }{ \rho_0 }.
\end{align*}
By definition of $\lambda_q$, we have almost surely
\beqeeq{
\EE*{ \mathcal V(\rho_{p+q}) }{ \rho_p } \leq \lambda_q \mathcal V(\rho_p)
}
hence
\beqeeq{
\frac{ \EE*{ \mathcal V(\rho_{p+q}) }{ \rho_p } }{ \mathcal V(\rho_p) } \leq \lambda_q.
}
Substituting this into the previous inequality gives
\begin{align*}
\EE{\mathcal V(\rho_{p+q})}{\rho_0} 
&\leq \lambda_q \EE*{ \mathcal V(\rho_p) }{ \rho_0 } \\
&\leq \lambda_q \lambda_p \mathcal V(\rho_0).
\end{align*}
Taking the supremum over $\rho_0\in\S_{\mix}$ proves the result.
\end{proof}

\begin{corollary}
    The sequence $(\gamma_p)_{p\ge 1}$ is super-additive, that is,
    \beqeeq{
        \gamma_{p+q} \,\ge\, \gamma_p + \gamma_q
        \qquad \text{for all } p,q \ge 1.
}
    As a consequence, the following limit exists and satisfies
    \beqeeq{
        \lim_{p\to\infty} \frac{\gamma_p}{p} 
        \;=\; \sup_{p\ge 1} \frac{\gamma_p}{p}.
}
\end{corollary}

This limit characterizes, in an asymptotic sense, an upper bound on the exponential decay rate
associated with grouping the dynamics into blocks of increasing length.




\section{Exponential stability of estimated quantum trajectories}\label{sec:stability}
Closely related to the question of purification is the notion of asymptotic stability of estimated quantum trajectories (or quantum filters). The problem considered is as follows.
Suppose that the initial state of the system is unknown, but that we still have access to the measurement outcomes. In this case, we cannot construct the trajectory $(\rho_n)$ since the starting point is missing. However, we can construct an estimation of the trajectory by fixing an arbitrary initial state $\hat\rho_0$ and making it evolve as if it were the true state. More precisely, we define the estimated trajectory by

\beqeeq{
	\hat\rho_n := \normalized{V_{i_n}  \ldots V_{i_1} \hat\rho_0 V_{i_1}^\dag \ldots V_{i_n}^\dag}
}
when $(i_1,\ldots,i_n)$ has been observed. The natural question that arises is whether this estimated trajectory converges to the true trajectory as $n$ increases.

First, to ensure that the estimated trajectory is well defined at all times, one typically requires that the initial states satisfy $\supp(\rho_0) \subset \supp(\hat\rho_0)$\,---\,we write this as $\rho_0 \ll \hat\rho_0$\,---\,so that the denominator $\P_{\rho_0}$--a.s.\ never vanishes. Indeed when such hypothesis holds, there exists a constant $c>0$ such that $\rho_0 \leq c \hat\rho_0$ and then for all $(i_1,\ldots,i_n) \in \I^n,$
\beqeeq{
	\tr{V_{i_n} \ldots V_{i_1} \rho_0 V_{i_1}^\dag \ldots V_{i_n}^\dag} \leq c \tr{V_{i_n} \ldots V_{i_1} \hat\rho_0 V_{i_1}^\dag \ldots V_{i_n}^\dag}
}
which implies that
\beqeeq{
	\P_{\rho_0}\oo\left\{\tr{V_{i_n} \ldots V_{i_1} \hat\rho_0 V_{i_1}^\dag \ldots V_{i_n}^\dag} = 0 \right\} \cc= 0.
}
A convenient way to guarantee that $\rho_0 \ll \hat\rho_0$ for any initial state $\rho_0$ is to choose the initial estimate $\hat\rho_0$ to be full rank.\\

To quantify the proximity between the true state $\rho_n$ and its estimate $\hat\rho_n$, we use the quantum fidelity \cite{jozsa1994fidelity}, a standard measure of proximity between quantum states. For density operators \( \rho, \sigma \in \mathcal{S} \), it is defined as
\beqeeq{
  F(\rho, \sigma) = \bigl(\operatorname{Tr}\sqrt{\sqrt{\rho}\,\sigma\sqrt{\rho}}\bigr)^2.
}
It satisfies \( 0 \le F(\rho, \sigma) \le 1 \), with equality \( F(\rho, \sigma) = 1 \) if and only if \( \rho = \sigma \),  
and \( F(\rho, \sigma) = 0 \) whenever $\supp(\rho)$ and $\supp(\sigma)$ are orthogonal.  
In the particular case where one of the two states is pure, say \( \rho = \ketbra{\psi} \), the expression simplifies to
\beqeeq{
  F(\ketbra{\psi}, \sigma) = \bra{\psi}\sigma\ket{\psi}.
}
The next theorem establishes the exponential stability of the filter.

\begin{theorem}[Exponential stability]\label{thm:exponential_stability} Assume that $\rho_0 \ll \hat\rho_0$ and that \pur holds.
 Then the estimated trajectory converges exponentially fast to the true trajectory: for all $n\geq 0,$ \beqeeq{\E{1- F(\rho_n,\hat\rho_n)} \leq C e^{- \gamma \floor{n/\bar p}}}
where $C=\bigl\|\hat\rho_0^{-\frac12}\rho_0\hat\rho_0^{-\frac12}\bigr\|_{\infty} \mathcal V(\hat \rho_0)$ and $\gamma$ is the rate of purification of Theorem \ref{thm:exponential_purification}. 
\end{theorem}
Here $\|\cdot\|_\infty$ denotes the operator norm (largest singular value), and $\hat\rho_0^{-1/2}$ is the inverse of $\hat\rho_0^{1/2}$ on $\supp(\hat\rho_0)$,
extended by zero elsewhere.\\

{
As in the previous section, we also obtain the following tail bound.
\begin{corollary}[Exponential fidelity tail bound]
Under the assumptions of Theorem~\ref{thm:exponential_stability}, for all $\varepsilon>0$ and all $n\ge 0$,
\beqeeq{
\mathbb P\big(1 - F(\rho_n,\hat\rho_n) \ge \varepsilon\big)
\le
\frac{C}{\varepsilon}\, e^{-\gamma\lfloor n/\bar p\rfloor}.
}
\end{corollary}

\begin{proof}
The proof is the same as that of Corollary~\ref{cor:exponential_tail_bound_purity}, replacing Theorem~\ref{thm:exponential_purification} with Theorem~\ref{thm:exponential_stability}.
\end{proof}
}

Let us now proceed with the proof of the theorem.

\begin{proof}[Proof of Theorem \ref{thm:exponential_stability}]
We start by recalling that, for a given initial state $\rho$, the probability of observing a sequence $I\in\I^n$ of measurement outcomes is 
\beqeeq{
\P_\rho(I)=\tr{V_{I}\rho V_{I}^\dag}.
}
We denote by $\E[\rho]{\cdot}$ the expectation with respect to this probability measure.  
For two different initial states $\rho_0^{(a)}$ and $\rho_0^{(b)}$, the corresponding quantum trajectories are
\beqeeq{
\rho_n^{(a)}(I)
    := \frac{V_{I}\,\rho_0^{(a)}\,V_{I}^\dag}{\tr{V_{I}\rho_0^{(a)}V_{I}^\dag}},
    \qquad
\rho_n^{(b)}(I)
    := \frac{V_{I}\,\rho_0^{(b)}\,V_{I}^\dag}{\tr{V_{I}\rho_0^{(b)}V_{I}^\dag}}.
}
A straightforward computation gives
\beqeeq{
\E[\rho_0^{(a)}]{1-\tr{\rho_n^{(a)}\rho_n^{(b)}}}
=1-\sum_{I\in\I^n}\tr*{V_{I}\rho_0^{(a)}V_{I}^\dag
\frac{V_{I}\rho_0^{(b)}V_{I}^\dag}{\tr{V_{I}\rho_0^{(b)}V_{I}^\dag}}},
}
which shows in particular that this expectation is linear in the argument $\rho_0^{(a)}$. 

Since $\supp(\rho_0) \subset \supp (\hat\rho_0)$, we can decompose $\hat\rho_0$ as
\beqeeq{
\hat\rho_0 = \mu\, \rho_0 + (1-\mu)\, \rho_0^c,
}
for some complementary state $\rho_0^c \in \mathcal S$ and a constant $\mu>0,$ given by $\mu = \bigl\|\hat\rho_0^{-1/2}\,\rho_0\,\hat\rho_0^{-1/2}\bigr\|_{\infty}^{-1}.$

Let us now choose $\rho_0^{(a)} = \hat \rho_0$ and expand by linearity. We obtain
\beqeeq{
	\E[\hat\rho_0]{1 - \tr{\hat\rho_n \rho_n^{(b)}} } = \mu \E[\rho_0]{1 - \tr{\rho_n \rho_n^{(b)}}}   + (1-\mu) \E[\rho_0^c]{1 - \tr{\rho_n^c \rho_n^{(b)}}   }.
}
Let us set $\rho_0^{(b)} = \hat\rho_0$  as well. It follows
\begin{equation}\label{eq:decomposition_with_mu}
    \E[\hat\rho_0]{1 - \tr{\hat\rho_n^2} } = \mu \E[\rho_0]{1 - \tr{\rho_n \hat\rho_n}}   + (1-\mu) \E[\rho_0^c]{1 - \tr{\rho_n^c \hat\rho_n}   }.
\end{equation}
Since for any state $\rho,$ we have   $0 \leq \tr{\rho^2} \leq 1,$ it follows that $1 - \tr{\rho^2} \leq \sqrt{1 - \tr{\rho^2}}$ and then
\beqeeq{
	 \E[\hat\rho_0]{1 - \tr{\hat\rho_n^2}}   \leq    \E[\hat\rho_0]{ \sqrt{1-\tr{\hat\rho_n^2}}}.
}
As the expectation is taken with respect to $\P_{\hat\rho_0}$, Theorem~\ref{thm:exponential_purification} applies directly to the right-hand side, yielding
\beqeeq{
	\E[\hat\rho_0]{1 - \tr{\hat\rho_n^2}}   \leq V(\hat\rho_0) e^{-\gamma \floor{{n}/{\bar p}}}.
}
Using Equation~\eqref{eq:decomposition_with_mu} and the fact that 
$(1-\mu)\E[\rho_0^c]{1 - \tr{\rho_n^c \hat\rho_n}} \geq 0$, we obtain
\beqeeq{
    \E[\rho_0]{1 - \tr{\rho_n \hat\rho_n}} \leq \mu^{-1} V(\hat\rho_0)\, e^{-\gamma \lfloor n/\bar p \rfloor}.
}
Then using the well-known inequality $F(\rho,\sigma) \geq \tr{\rho \sigma}$ we get the theorem, with constant $ C = \mu^{-1} V(\hat\rho_0).$
\end{proof}

We now turn to the complementary situation in which
\pur does not hold.
In this case, the estimated trajectory may fail to converge to the true state. 

\begin{proposition}[Constancy of the fidelity on a dark subspace]
Let $D \subset \H$ be a dark subspace with projection $\pi_D$, in the sense that
for every finite record $I=(i_1,\dots,i_p)$,
\beqeeq{
    \pi_D\,V_I^\dagger V_I\,\pi_D = \mu_I\,\pi_D,
}
for some scalar $\mu_I.$
Assume that $\rho_0$ and $\hat\rho_0$ are both supported on $D,$ with $\rho_0 \ll \hat\rho_0.$
Then the fidelity is conserved: for all $n\geq 0,$
\beqeeq{
F(\rho_n,\hat\rho_n)=F(\rho_0,\hat\rho_0).
}
\end{proposition}

\begin{proof}
As observed in the proof of Proposition~\ref{prop:purity_constant_in_dark_subspace},
for any finite record $I$ there exists a unitary operator $U_I$ such that,
since both $\rho_0$ and $\hat\rho_0$ are supported on $D$,
\beqeeq{
\rho_n = U_I \rho_0 U_I^\dagger,
\qquad
\hat\rho_n = U_I \hat\rho_0 U_I^\dagger.
}
The result follows from the invariance of fidelity under simultaneous unitary conjugation.
\end{proof}

\section{Example}
\label{sec:simulations}

We illustrate our theoretical results with a simple example that combines Hamiltonian and measurement-induced dynamics. We consider a chain of four spins, each associated with a two-dimensional Hilbert space \(\mathbb{C}^2\). The global Hilbert space is given by
$
\mathcal{H} = (\mathbb{C}^2)^{\otimes 4}.
$
The system evolves under a Hamiltonian with Ising exchange interactions and a uniform external magnetic field $\vec B = (B_x, 0, B_z)$ whose direction lies in the $x$–$z$ plane. Projective measurements of the Pauli operator $\sigma_z$ are repeatedly performed on the last spin of the chain, see Fig.\ \ref{fig:spin_chain}.
\begin{figure}[htpb]
    \centering
    \hspace*{2.5cm}
    \begin{tikzpicture}[
        >=Latex,
        line cap=round,
        spin/.style={ball color=white, circle, draw=black, thick, minimum size=9mm, inner sep=0pt},
        link/.style={line width=1pt, draw=gray!70},
        Bvec/.style={->, line width=1.6pt, draw=blue!70!black},
        cavity/.style={draw=red!70!black, fill=red!5, thick},
        field/.style={fill=red!20, draw=none, opacity=0.25},
        Blabel/.style={blue!70!black, font=\normalsize},
        spinarrow/.style={->, line width=1pt, draw=black}
        ]
        \node[spin] (s2) at (2,0) {};
        \node[spin] (s3) at (4,0) {};
        \node[spin] (s4) at (6,0) {};
        \node[spin] (s5) at (8,0) {};
        \draw[link] (s2) -- (s3);
        \draw[link] (s3) -- (s4);
        \draw[link] (s4) -- (s5);
        \draw[spinarrow] ($(s2.center)+(-105:6.5mm)$) -- ($(s2.center)+(75:6.5mm)$);
        \draw[spinarrow] ($(s3.center)+(-120:6.5mm)$) -- ($(s3.center)+(60:6.5mm)$);
        \draw[spinarrow] ($(s4.center)+(-100:6.5mm)$) -- ($(s4.center)+(80:6.5mm)$);
        \draw[spinarrow] ($(s5.center)+(-90:6.5mm)$) -- ($(s5.center)+(90:6.5mm)$);
        \coordinate (Bstart) at (1,-1.5);
        \draw[Bvec] (Bstart) -- ++(1.5,1.0) node[above right, Blabel] {$\vec B$};
        \def\cavh{0.3}
        \def\cavw{1.0}
        \def\cavy{0.8}
        \draw[cavity] ($(s5)+(0,\cavy)$) rectangle ++(\cavw, \cavh);
        \draw[cavity] ($(s5)+(0,\cavy)$) rectangle ++(-\cavw, \cavh);
        \draw[cavity] ($(s5)+(0,-\cavy)$) rectangle ++(\cavw, -\cavh);
        \draw[cavity] ($(s5)+(0,-\cavy)$) rectangle ++(-\cavw, -\cavh);
        \shade[
            draw=none,
            inner color=red!35,
            outer color=red!5,
            opacity=0.35
            ]
        ($(s5)+(0,0.35)$) ellipse [x radius=1.1, y radius=0.45];
        \shade[
            draw=none,
            inner color=red!35,
            outer color=red!5,
            opacity=0.35
            ]
        ($(s5)+(0,-0.35)$) ellipse [x radius=1.1, y radius=0.45];
        \draw[thick]
        ($(s5)+(2.2,0.9)$) -- ($(s5)+(1.2,0.2)$);
        \node[right] at ($(s5)+(2.2,0.9)$)
        {\small measurement};
    \end{tikzpicture}
    \caption{
        Spin chain subject to a uniform magnetic field $\vec B$, with repeated projective measurements performed on the last spin along the $z$ axis.
        }
    \label{fig:spin_chain}
\end{figure}
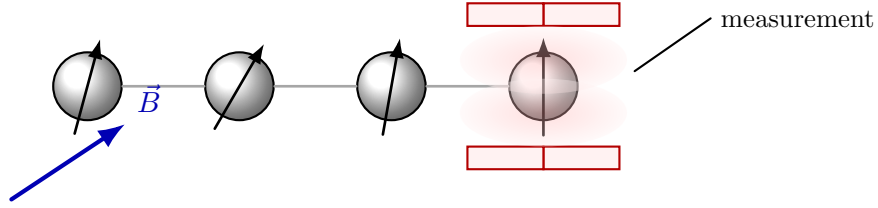

\noindent The Hamiltonian reads:
\beqeeq{
    H = -J \sum_{j=1}^{3} \sigma_j^z \sigma_{j+1}^z - B_x \sum_{j=1}^4 \sigma_j^x - B_z \sum_{j=1}^4 \sigma_j^z,
}
where \(J \in \mathbb{R}\) denotes the coupling strength, \(B_x, B_z\) are the components of the uniform magnetic field, and \(\sigma_j^\alpha\) denotes the Pauli matrix \(\sigma^\alpha\) acting on the \(j\)-th site:
\beqeeq{
    \sigma_j^\alpha = \mathbb I_{\C^2}^{\otimes (j-1)} \otimes \sigma^\alpha \otimes \mathbb I_{\C^2}^{\otimes (4-j)}, \quad \alpha \in \{x, z\}
}
with $\mathbb I_{\C^2}$ the identity operator on $\C^2.$
We let the system evolve according to this Hamiltonian for a duration $\tau,$ then perform a measurement on the last spin along the $z$-axis. The unitary operator corresponding to the Hamiltonian evolution is
$U = e^{-i \tau H}.$
The projections corresponding to the measurement of the last spin are
\beqeeq{
P_0 = \mathbb I_{\C^2}^{\otimes 3} \otimes \ket{0}\bra{0}, \qquad
P_1 = \mathbb I_{\C^2}^{\otimes 3} \otimes \ket{1}\bra{1}
}
which correspond respectively to measurement outcomes \(0\) and \(1\) in the eigenbasis of \(\sigma_z\).
The associated Kraus operators are 
\beqeeq{
V_0 =  P_0 U, \qquad \qquad
V_1 =  P_1 U.
}
Iterating this procedure of Hamiltonian evolution for a duration $\tau$ followed by a measurement on the last spin, the system evolves according to a quantum trajectory dynamics given by
\beqeeq{
	\rho_{n+1} = \normalized{V_{i} \rho_n V_i^\dag} \qquad \text{with probability } \tr{V_i \rho_n V_i^\dag},~~i=0,1 
}
where, as before, the index $n$ labels the measurement times.

We perform numerical simulations of this model, choosing parameters $J=1,$ $\tau=1,$ $B_x = 1,$  $B_z=1,$ and setting the initial state to be the completely mixed state $\rho_0 = \mathbb I_{\C^2}^{\otimes 4}/2^4.$ 
In Fig.~\ref{fig:purity}(a) we show the normalized quantity
\beqeeq{
    \widetilde{\mathcal V}(\rho_n) := {\mathcal V(\rho_n)}/{\mathcal V(\rho_0)}
}
for individual sample trajectories (thin light-blue curves), together with their empirical mean (thick red curve),
{which corresponds to the quantity
\begin{equation}
\frac{1}{N}\sum_{k=1}^{N}\V(\rho_n^{(k)})
\approx
\E{\V(\rho_n)},
\qquad N=300,
\end{equation}
where $\rho_n^{(k)}$ denotes the $k$-th simulated trajectory.} There is a clear trend towards purification over time. Figure~\ref{fig:purity}(b) shows the same data on a log-linear scale to emphasize the exponential convergence. It also shows (dashed-blue lines) the upper bounds from Equation~\eqref{eq:exp_bound_with_arbitrary_p} for different values of the time interval $p$, demonstrating the improvement {of the bound} as $p$ increases.

\begin{figure}[htpb]
    \centering
    \begin{minipage}[t]{0.48\textwidth}
        \vspace{0pt}
        \centering
        \makebox[\linewidth]{\textbf{(a) Impurity as a function of time}}\\[0.3em]
        \includegraphics[width=\linewidth]{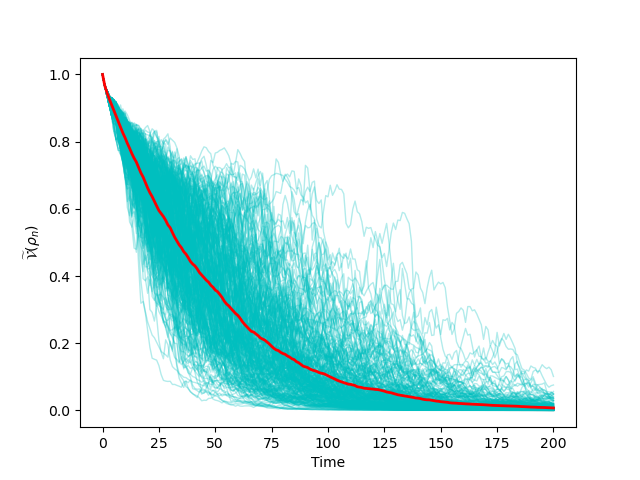}
    \end{minipage}
    \hfill
    \begin{minipage}[t]{0.48\textwidth}
        \vspace{0pt}
        \centering
        \makebox[\linewidth]{\textbf{(b) Log-impurity as a function of time}}\\[0.3em]
        \includegraphics[width=\linewidth]{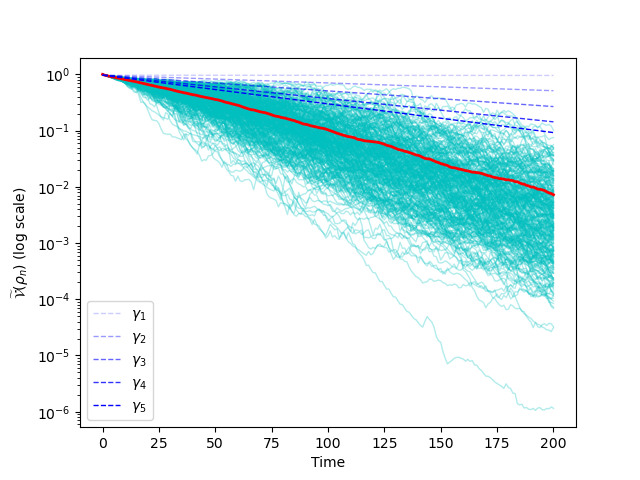}
    \end{minipage}
    \caption{
        (a) 300 trajectories of $\widetilde{\mathcal V}(\rho_n)= {\mathcal V(\rho_n)}/{\mathcal V(\rho_0)}$ (light blue) and the corresponding exponential decrease of the empirical average (thick red).
        (b) Logarithmic-scale plot of 300 trajectories of $\widetilde{\mathcal V}(\rho_n)$ (light blue) together with the empirical average (thick red) showing the expected linear behavior. Theoretical upper bounds $e^{-\gamma_p \floor{n/p}}$ stated in Equation \eqref{eq:exp_bound_with_arbitrary_p} are plotted in blue dashed line for different values of $p=1,\dots , 5$.}
    \label{fig:purity}
\end{figure}

In this model, one-step darkness holds, for instance for the subspaces $\range(U^\dagger P_0 U)$ and $\range(U^\dagger P_1 U)$, meaning that $\gamma_1 = 0$. However, the numerics show that the darkness property is lost from $p=2$ onward, and $\gamma_2 > 0$. { Moreover, we find $\pbar = 7$ (and $\dim \mathcal O = 128$), showing that the darkness property can be broken strictly before $\pbar$. These values hold for all parameters in the considered range.}

In Fig.~\ref{fig:gamma2-jtau}(a), we show the quantity $\gamma_2(J,\tau)/(2\tau),$  {where $\gamma_2(J,\tau)$ is computed as defined in Equation \eqref{eq:def_gamma_p},
as a function of the parameters $J$ and $\tau$, for fixed $B_x=1$ and $B_z=1$. We divide by $\tau$ in order to express the convergence rate in terms of physical time
rather than number of steps.} In Fig.~\ref{fig:gamma2-jtau}(b), we show in a similar manner the empirical rate $\gamma_{\rm emp}(J,\tau)/\tau$, obtained from the effective decay rate of the red curve in Fig.~\ref{fig:purity}(b) {(computed by averaging over 300 simulated trajectories with the corresponding parameters)}. {The numerical results indicate that the theoretical bound and the empirical rate display similar qualitative behavior over this range of parameters (for example, the maximum is reached approximately in the same region).}
Figure~\ref{fig:gamma2-bb} shows a similar dependence but with respect to parameters $B_x,B_z$, for fixed $J=1,~\tau=1$.

\begin{figure}[htpb]
    \centering
    \begin{minipage}{0.48\textwidth}
        \centering
        \makebox[\linewidth]{\textbf{(a) Theoretical bound $\gamma_2/(2\tau)$}}\\[0.3em]
        \includegraphics[width=\linewidth]{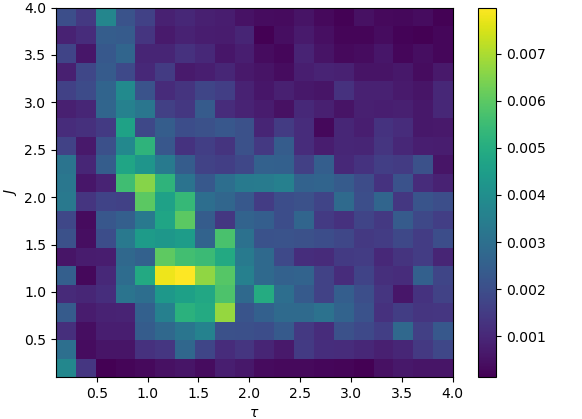}
    \end{minipage}
    \hfill
    \begin{minipage}{0.48\textwidth}
        \centering
        \makebox[\linewidth]{\textbf{(b) Empirical rate $\gamma_\emp/\tau$ }}\\[0.3em]
        \includegraphics[width=\linewidth]{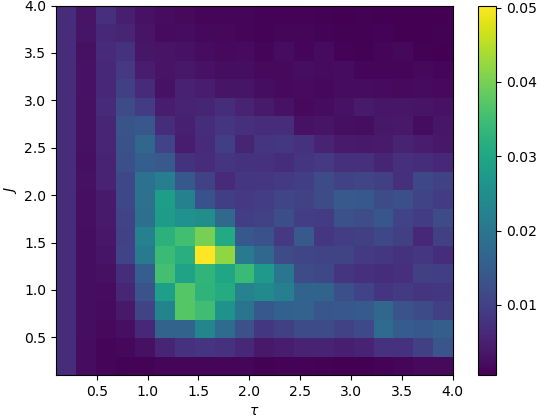}
    \end{minipage}
    \caption{
        (a) The rate bound $\gamma_{\textrm{2}}/(2\tau)$ as a function of $J$ and $\tau$, for $B_x=1$ and $B_z=1$. 
        (b) The empirical rate $\gamma_{\textrm{emp}}/\tau$ as a function of $J$ and $\tau$, for $B_x=1$ and $B_z=1$.}
    \label{fig:gamma2-jtau}
\end{figure}

\begin{figure}[t]
    \centering
    \begin{minipage}{0.48\textwidth}
        \centering
        \makebox[\linewidth]{\textbf{(a) Theoretical bound $\gamma_2/2$}}\\[0.3em]
        \includegraphics[width=\linewidth]{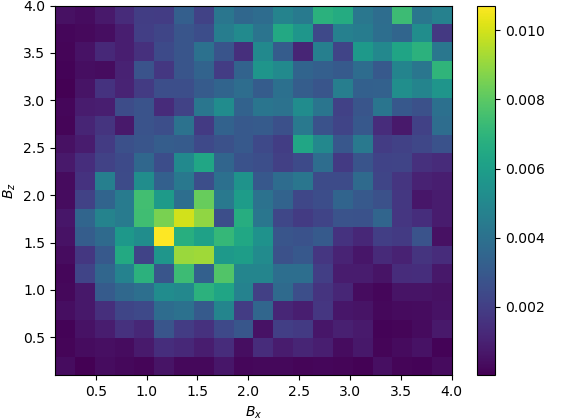}
    \end{minipage}
    \hfill
    \begin{minipage}{0.48\textwidth}
        \centering
        \makebox[\linewidth]{\textbf{(b) Empirical rate $\gamma_\emp$}}\\[0.3em]
        \includegraphics[width=\linewidth]{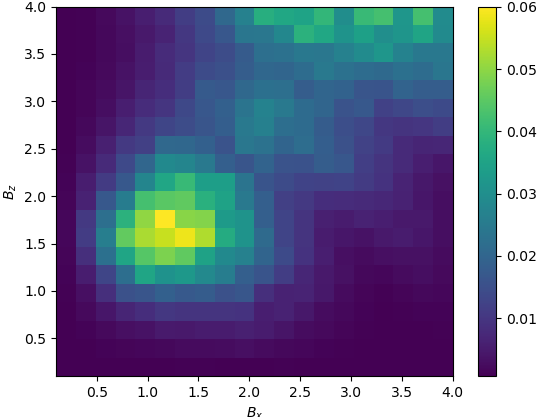}
    \end{minipage}
    \caption{
        (a) The rate bound $\gamma_{\textrm{2}}/2$ as a function of $(B_x,B_z)$, for $J=1$ and $\tau=1$.
        (b) The empirical rate $\gamma_{\textrm{emp}}$ as a function of $(B_x,B_z)$, for $J=1$ and $\tau=1$.
        }
    \label{fig:gamma2-bb}
\end{figure}

\section{Conclusions}\label{sec:conc}

{

 In this paper, we study the behavior of the purity of discrete-time quantum trajectories under perfect measurements, assuming the absence of dark subspaces. An alternative proof of the Kümmerer--Maassen purification theorem is first provided based on Lyapunov techniques, yielding a direct argument for almost sure purification via a simple impurity-based Lyapunov function.

 This asymptotic result is then strengthened by showing that purification occurs exponentially fast in expectation, providing a quantitative description of the speed of convergence towards pure states. The analysis further shows that this exponential behavior is driven by a finite-step contraction mechanism, and that it suffices to detect the absence of dark subspaces up to a finite length.

 The state estimation problem is next considered, and exponential stability of the associated quantum filter is established. When the true trajectory and an estimated trajectory are driven by the same measurement record, the estimation error decreases exponentially fast in expectation. This provides a quantitative version of filter stability and shows how purification translates into a progressive loss of dependence on the initial estimate.

 Finally, the theory is illustrated using a monitored spin-chain model combining Hamiltonian evolution and repeated local measurements. The numerical results are consistent with the theoretical predictions, showing clear exponential decay of the impurity and improved bounds when longer blocks of observations are considered.

}

A natural continuation of this work would be to extend this analysis to continuous-time quantum trajectories and to infinite dimensional systems. In the latter setting, non-purification of quantum trajectories is not necessarily equivalent to the presence of dark subspaces, as recently emphasized in~\cite{girotti2025purification}.

Finally, in the context of measurement-induced phase transitions in random hybrid circuits (see~\cite{fisher2023random} for a review), it would be interesting to investigate how the purification rate scales with the system size and the measurement frequency, and to explore potential
connections with the dynamical purification phase transition of Ref.~\cite{gullans2020dynamical}.

\vspace{4mm}

\noindent
{\bf Acknowledgments:} This work was supported by the Engineering and Physical Sciences Research Council, grants no. EP/T022140/1 and EP/V031201/1. N. A. acknowledges support from ANR-19-CE48-0003 and ANR-21-CE47-0015.

\bibliographystyle{iopart-num}
\bibliography{bibliography-01012026}

\appendix
\renewcommand{\thetheorem}{A.\arabic{theorem}}
\setcounter{theorem}{0}
\section*{Appendix}

The following theorem is a slight adaptation of \cite[Theorem~1, Chapter~8]{kushner1971introduction} (see also \cite[Theorem A.1]{amini2011design}).

\begin{theorem}\label{thm:accumulation_points_are_zeros_of_the_increment}
Let \( X_k \) be a  Markov chain in the compact state space \( S \). Suppose that there exists a nonnegative function \( \mathcal V(X) \) satisfying, for some integer $p\geq 1,$
\begin{equation}\label{eq:inc}
\mathbb{E} \left[ \mathcal V(X_{k+p}) \mid X_k \right] - \mathcal V(X_k) \leq \mathcal Q(X_k),
\end{equation}
where \( \mathcal Q(X) \) is a nonpositive continuous function of \( X \). Then the \( \omega \)-limit set \( \Omega \) of \( X_k \) (i.e. the set of accumulation points of the sample path) is almost surely included in the following set:
\beqeeq{
I := \{ X\in S \mid \mathcal Q(X) = 0 \}.
}
\end{theorem}
\begin{proof}
For the increment $p=1,$ this is a well-known result (see, e.g., \cite{kushner1971introduction}). Here we show that the same conclusion holds for any $p\geq 1.$ For each $0\leq r\leq p-1,$ define the subsequence $M_n^r:=X_{np+r}.$ It can be easily checked that each $M_n^r$ is also a Markov chain. Now take $k=np+r,$ then Equation \eqref{eq:inc} can be rewritten as 
\begin{equation*}
\mathbb{E} \left[ \mathcal V(M_{n+1}^r) \mid M_n^r \right] - \mathcal V(M_n^r) \leq \mathcal Q(M_n^r).
\end{equation*}
Now we can apply the result in \cite{kushner1971introduction} to conclude that $\Omega(M^r)\subset I$ for each $r.$ To complete the argument, we now show that the $\omega$-limit set of $X$ satisfies
\beqeeq{
\Omega(X)=\bigcup_{r=0}^{p-1}\Omega(M^{r}).
}
We establish this identity by proving the two inclusions separately.

Let $x \in \Omega(X)$. Then there exists an increasing sequence of indices 
$k_n \to \infty$ such that $X_{k_n} \to x$.  
Write each $k_n$ in the form
\beqeeq{
k_n = p\, m_n + r_n, \qquad r_n \in \{0,\dots,p-1\}.
}
Since there are only finitely many possible residues, one value 
$r \in \{0,\dots,p-1\}$ occurs infinitely often in the sequence $(r_n)$.  
Restricting to this infinite subsequence, we obtain
\beqeeq{
X_{k_n} = X_{p\, m_n + r} = M^{r}_{m_n} \longrightarrow x.
}
Hence $x \in \Omega\!\oo M^{r}\cc$, and therefore
$
\Omega(X) \subseteq \bigcup_{r=0}^{p-1} \Omega\!\oo M^{r}\cc.
$

Conversely, if $x \in \Omega\!\oo M^{r}\cc$ for some $r$, then there exists 
$m_n \to \infty$ such that $M^{r}_{m_n} \to x$.  
Since $M^{r}_{m_n} = X_{p\, m_n + r}$, it follows that $x$ is also an accumulation point 
of the full sequence $(X_k)$. Thus,
$
\bigcup_{r=0}^{p-1} \Omega\!\oo M^{r}\cc \subseteq \Omega(X).
$
Combining the two inclusions gives the desired equality.
Hence,  $\Omega\subset\bigcup_{r=0}^{p-1} I=I,$ which finishes the proof.
\end{proof}

\end{document}